\documentclass[12pt]{article}

\usepackage{amsmath,amsthm,amsfonts,amsxtra,amssymb,upref}
\usepackage[mathscr]{eucal}
\usepackage[a4paper,margin=2.5cm]{geometry}

\usepackage{mathtools}
\usepackage{enumerate}
\usepackage{color}

\usepackage[parfill]{parskip}    % Activate to begin paragraphs with an empty line rather than an indent

\numberwithin{equation}{section}

\theoremstyle{plain}
\newtheorem{thm}{Theorem}[section]
\newtheorem{cor}[thm]{Corollary}
\newtheorem{lem}[thm]{Lemma}
\newtheorem{prop}[thm]{Proposition}

\theoremstyle{definition}
\newtheorem{assumption}[thm]{Assumption}

\newtheorem{rem}[thm]{Remark}

\newcommand{\eps}{\varepsilon}
\newcommand{\N}{\mathbb{N}}
\newcommand{\R}{\mathbb{R}}

\newcommand{\U}{\mathcal{U}}

\newcommand{\LL}{\mathcal{\big\langle}}
\newcommand{\RR}{\mathcal{\big\rangle}}

\newcommand{\cH}{\mathcal{H}}
\newcommand{\jap}[1]{\langle{#1}\rangle}

\definecolor{ddcyan}{rgb}{0,0.6,1.0}
\DeclareMathOperator{\im}{Im}

\DeclareMathOperator{\re}{Re}

\newcommand{\abs}[1]{\lvert{#1}\rvert}

\newcommand{\norm}[1]{\lVert{#1}\rVert}
\newcommand{\ket}[1]{\lvert{#1}\rangle}
\newcommand{\bra}[1]{\langle{#1}\rvert}
\newcommand{\ip}[2]{\langle{#1},{#2}\rangle}

\newcommand{\cO}{\mathcal{O}}

\DeclareMathOperator{\sign}{sign}

\newcommand{\bR}{\mathbf{R}}
\newcommand{\bC}{\mathbf{C}}
\newcommand{\bone}{\mathbf{1}}

\newcommand{\cB}{\mathcal{B}}
\newcommand{\cK}{\mathcal{K}}
\newcommand{\cU}{\mathcal{U}}
\newcommand{\We}{W_{\eps}}
\newcommand{\we}{w_{\eps}}
\newcommand{\wes}{w_{\eps}^{\ast}}
\newcommand{\Reff}{R_{\textrm{eff}}}
\newcommand{\Heff}{H_{\textrm{eff}}}
\newcommand{\tb}{\tilde{b}}
\newcommand{\Deps}{D_{\nu,\eps}}

\newcommand{\Esf}{\mathsf{E}}

\newcommand{\Ie}{I_{\eps}}
\newcommand{\xe}{x_0(\eps)}

\definecolor{mygreen}{rgb}{0.0,0.5,0.0}

\newcommand{\cV}{\mathcal{V}}

\title{A modified Fermi Golden Rule at threshold for 3D magnetic Schr\"odinger operators}

\author{Pavel Exner\footnote{Doppler Institute for Mathematical Physics and Applied Mathematics, Prague  Czechia, and  Department of Theoretical Physics, Nuclear Physics Institute, Czech Academy of Sciences,
\v Re\v z near Prague, Czechia, \texttt{pavel.exner@ujf.cas.cz}}, \
Arne Jensen\footnote{Department of Mathematical Sciences, Aalborg
University, Thomas Manns Vej 23, DK-9220 Aalborg~\O{}, Denmark,
\texttt{matarne@math.aau.dk}} \ and  Hynek Kova\v{r}\'{\i}k\footnote{DICATAM, Sezione di Matematica, Università degli studi di Brescia, Via Branze 38, Brescia, 25123, Italy, \texttt{hynek.kovarik@unibs.it}}}

\begin{document}
\allowdisplaybreaks

\maketitle


\smallskip


\begin{abstract}
In this paper we consider three-dimensional Schr\"odinger operators with a simple threshold eigenvalue. We show, under certain
assumptions, that when a small magnetic field is introduced, this eigenvalue turns into a resonance
in the time-dependent sense. We find the leading term in the asymptotic expansion of the imaginary part of the resonance and discuss
the principal differences with respect to resonances induced by  weak electric fields obtained previously in the literature.
\end{abstract}

%%%%%%%%%%%%%%%%%%%%%%
\section{Introduction}

We consider the following operators on $\cH=L^2(\bR^3)$:
 %----------------%
\begin{align}
H_0&=-\Delta+V,\label{H0}
\\
H_{\lambda}&=(P-\lambda A)^2+V,%\label{H-lambda}
\notag
\end{align}
 %----------------%
where $P=-i\nabla$ and $V\colon\bR^3\to\bR$ is a bounded potential that decays sufficiently rapidly, say $\abs{V(x)}\leq C(1+\abs{x})^{-\delta}$ with $\delta>2$. The magnetic potential $A\colon \bR^3\to\bR^3$, which generates the magnetic field $B=\nabla \times A$, is also assumed to decay sufficiently rapidly, see Assumption \ref{assump}(2) below for details. Under these  assumptions  the point zero in the spectrum of $H_0$ may be an eigenvalue, a zero resonance, both, or neither; by zero resonance we mean that there exists a solution to $H_0\psi=0$ (in the sense of distributions), which satisfies  $\psi  \in L^2(\bR^d, \jap x^{ -\frac 12-\delta}\, dx)\setminus L^2(\bR^3)$ for any $\delta>0$.

In this paper we assume that zero, the essential spectrum threshold, is a simple eigenvalue of $H_0$, but not a resonance.
Suppose that we turn on a magnetic field, i.e.~consider $H_{\lambda}$ for $\lambda$ small. The question that we address is the following:

What happens to the eigenvalue embedded at the threshold zero? The expectation is that it either becomes an isolated eigenvalue close to zero or becomes a resonance. Notice that under our decay conditions on $V$ and $B$ the operator $H_\lambda$ cannot have an eigenvalue embedded in the continuous spectrum \cite{ahk}. Assume that $\psi_0\in\cH$, $\norm{\psi_0}=1$, satisfies $H_0\psi_0=0$. We impose a condition that ensures the eigenvalue becomes a resonance, that is, a metastable state of the system. A time-dependent description of a resonance uses the survival probability of $e^{-itH_{\lambda}}\psi_0$ in the state $\psi_0$, cf.~\cite{Ex85, hun, JN, ort, sw}
and the review paper \cite{har}.
 We are going to show that
 %----------------%
\begin{equation}\label{exp-decay}
\ip{\psi_0}{e^{-itH_{\lambda}}\, \psi_0}=e^{-it(x_0(\lambda)-i\Gamma(\lambda))}+\cO(\lambda^{2p})
\quad\text{as $\:\lambda\to  0$}
\end{equation}
 %----------------%
with the error term uniform in $t>0$. We have either $p=1/2$ or $p=2$, depending on properties of $H_0$ and $\psi_0$.
The point $x_0(\lambda)-i\Gamma(\lambda)$ with $\Gamma(\lambda)>0$ indicates the location of the resonance. Its imaginary part $\Gamma(\lambda)$, called the resonance width, is of main interest for us since it determines, up to the error term, the rate of decay of the survival probability $| \LL \psi_0, \, e^{-it H_\lambda}\, \psi_0 \RR|^2$, conventionally used to characterize the resonance lifetime.
We will find the leading terms in the asymptotics of $x_0(\lambda)$ and $\Gamma(\lambda)$.

Had operator $H_0$ an eigenvalue embedded in the continuum, $(0,\infty)$, then a perturbation $H_0+\lambda\widetilde{V}$ would turn it most often into a resonance the time evolution of which is described by \eqref{exp-decay}. In this case the imaginary part of the resonance position is typically given by the Fermi Golden Rule, $\Gamma(\lambda)=C_{\textup{FGR}}\lambda^2+\cO(\lambda^q)$ for some $q>2$, where the constant  $C_{\textup{FGR}}$ is generically positive.

Our result gives a modification of the Fermi Golden Rule in the situation when the eigenvalue is situated at the continuous spectrum threshold, and the perturbation is given by a magnetic field.
In particular, our main result, Theorem \ref{main}, states that
 %----------------%
\begin{equation} \label{gamma-0}
\Gamma(\lambda) = \Gamma_0\, |\lambda|^3 +  \mathcal{O}(\lambda^4) \quad\text{as} \ \lambda\to 0.
\end{equation}
 %----------------%
The coefficient  $\Gamma_0$ is non-negative. We would like to point out, however, that our main result, Theorem \ref{main}, provides an asymptotic expansion of $\Gamma(\lambda)$ even in the exceptional cases where $\Gamma_0=0$. See also Remark~\ref{finite-rank-remark} for further details.

%\AJnew{See comments on the above statements in my mail.}

The proof of our main result uses the factorization of the operator difference $H_\lambda-H_0$ recently developed in \cite{JKmag}. When combined with  the Schur-Lifschitz-Feshbach-Grushin (SLFG) formula, this factorization allows us to derive a threshold resolvent expansion of the operator $H_\lambda$ in the topology of weighted Sobolev spaces. Such a resolvent expansion then plays a crucial role in analyzing the location of the resonance, and hence in proving equation \eqref{gamma-0}.

It is illustrative to compare the behavior of the threshold resonances considered in this paper with those induced by a weak electric field, in other words arising from  a threshold eigenvalue as a result of an additive perturbation $\lambda W$. Such resonances were analyzed in \cite{JN}. There their existence was proved for $\lambda\to 0+$ under the condition $\LL \psi_0, W \psi_0\RR >0$, and it was shown that the leading term
of $\Gamma$ is generically of order $\lambda^{3/2}$, see \cite[Thm.~3.7]{JN}. A comparison with equation \eqref{gamma-0} thus reveals two important differences. First, the quantity $\Gamma(\lambda)$ associated to resonances induced by a magnetic field is independent of the sign of $\lambda$. Second, its size is much smaller  with respect to the case of a perturbation by an electric field which leads to a considerably  slower decay of the corresponding survival probability.

It should be also noticed that, by \cite[Prop.~1.1]{JN2},  a result of the type \eqref{exp-decay} defines the leading asymptotics in $x_0(\lambda)$ and $\Gamma(\lambda)$ uniquely, see Corollary \ref{cor-unique} for more details.

Some comments on the literature. Recall that in other systems the presence of the unperturbed eigenvalue at the continuum edge may also result in the coupling constant power in the leading term different from two, cf. \cite[Example 2]{Ho74} or \cite[Example 3.2.5]{Ex85} for the Friedrichs model example; in the words of J.~Howland the gold from which the rule is made is then mixed with brass. Another notable feature of the Fermi Golden Rule is that the constant $C_{\textup{FGR}}$ might vanish. Such situations were studied in \cite{cjn}, where the subsequent term in the asymptotic expansion  of $\Gamma(\lambda)$ was calculated. Threshold resonances arising from a perturbation of
the two-dimensional Pauli operator were found recently in \cite{BK}. In this case the asymptotic expansion of $\Gamma(\lambda)$ depends heavily on the flux of the magnetic field.

Our paper is organized as follows. In the next section we introduce the necessary notation and state the hypothesis on the magnetic and electric field. The essential resolvent expansions of the operator $H_\lambda$
are proved in Section \ref{sec-SLFG}, cf.~equation \eqref{eq-eff} and Lemmata \ref{lemma44}-\ref{lemma46}. The main result, Theorem \ref{main}, is stated and proved in Section \ref{sec-main}.

%%%%%%%%%%%%%%%%%%%%%%%%%%%%%%%%%%%%%%%%%%%%%%%%
\section{Notation, definitions, and assumptions}

The basic Hilbert space is $\cH=L^2(\bR^3)$. The norm is denoted by $\norm{\cdot}$.
Let $\mathcal{S}'(\bR^3)$ denote the tempered distributions. For $k,s\in\bR$ the weighted Sobolev spaces are defined by
 %----------------%
\begin{equation*}
H^{k,s}(\bR^3)=\{f\in\mathcal{S}'(\bR^3)
\mid (1+x^2)^{s/2}(1-\Delta)^{k/2}f\in L^2(\bR^3)\}
\end{equation*}
 %----------------%
with the norm
 %----------------%
\begin{equation*}
\norm{f}_{k,s}=\norm{(1+x^2)^{s/2}(1-\Delta)^{k/2}f}_{L^2(\bR^3)}.
\end{equation*}
 %----------------%
An equivalent norm is given by
 %----------------%
\begin{equation*}
\norm{f}^{'}_{k,s}=
\norm{(1-\Delta)^{k/2}(1+x^2)^{s/2}f}_{L^2(\bR^3)}.
\end{equation*}
 %----------------%
We abbreviate the notation to $H^{k,s}$.

We will frequently use the continuous embeddings
 %----------------%
\begin{equation}\label{embed}
H^{k_1,s_1} \hookrightarrow H^{k_2,s_2}, \quad k_1\leq k_2, \quad
s_1\leq s_2.
\end{equation}
 %----------------%
The inner product on $L^2(\bR^3)$, denoted by $\ip{\cdot}{\cdot}$, extends to a duality between $H^{k,s}$ and $H^{-k,-s}$.

The bounded operators between two Hilbert spaces $\cK_1$ and $\cK_2$ are denoted by $\cB(\cK_1,\cK_2)$. In particular, we sometimes use the notation
 %----------------%
\begin{equation*}
\cB(k_1,s_1;k_2,s_2)=\cB(H^{k_1,s_1},H^{k_2,s_2}).
\end{equation*}
 %----------------%

We impose the following assumptions on $V$ and $A$ in the operators given by \eqref{H0} and \eqref{Heps}.

\smallskip

 %----------------%
\begin{assumption}\label{assump}
\
\begin{enumerate}[(1)]
\item \label{assump1}
Assume $V\colon \bR^3\to\bR$ satisfies
\begin{equation*}
\abs{V(x)}\lesssim\jap{x}^{-\delta_1},\quad \delta_1>2.
\end{equation*}
\item \label{assump2}
Assume $B\in C^1(\bR^3;\bR^3)$ such that $\nabla\cdot B=0$. Assume $B=\nabla\times A$, $A\in C^2(\bR^3;\bR^3)$,  such that
\begin{equation}\label{decay}
\abs{B(x)}\lesssim \jap{x}^{-\delta_2-1},\quad
\abs{A(x)}\lesssim \jap{x}^{-\delta_2}, \quad \delta_2>2.
\end{equation}
\item \label{assump3}
Assume that $0$ is a simple eigenvalue of $H_0$, and furthermore that $H_0$ does not have a zero resonance. Let $\psi_0\in H^{2,0}$, such that $H_0\psi_0=0$ and $\norm{\psi_0}=1$. Assume $\psi_0$ is chosen to be \emph{real-valued}.
\end{enumerate}
\end{assumption}
 %----------------%

\medskip

\begin{rem}
For examples of operators $H_0$ satisfying condition (\ref{assump3}) we refer to \cite[Rem.~5.2]{JN} and \cite[Rem.~6.6]{JK}.
If $B$ satisfies the decay condition in \eqref{decay}, a construction of $A$ satisfying the corresponding decay condition is given in \cite{JKmag}.
\end{rem}

In the sequel it will be convenient  to work with the operator
 %----------------%
\begin{equation}\label{Heps}
H_{\eps}=(P-\sqrt{\eps}A)^2+V, \qquad \eps>0.
\end{equation}
 %----------------%
Our main result will then follow from a corresponding result for $H_\eps$ after identifying $\eps$ with $\lambda^2$, see the proof of Theorem \ref{main}.

We let
\begin{equation*}%\label{Weps-def}
\We=H_{\eps}-H_0=\sum_{j=1}^3(-\eps^{1/2}P_jA_j-\eps^{1/2}A_jP_j)
+\eps\abs{A}^2.
\end{equation*}
The following result follows directly from the assumptions and definitions. For completeness we give the proof.

\medskip

\begin{lem}
Let $A$ satisfy the decay condition in Assumption~\ref{assump}\textup{(\ref{assump2})}. Then we have
 %----------------%
\begin{equation}\label{Weps-map}
\We\in\cB(H^{1,s},H^{-1,s+\delta_2})\quad \text{for all $s\in\bR$}.
\end{equation}
 %----------------%
\end{lem}
\begin{proof}
Let $s\in\bR$ be given. Due to \eqref{decay} multiplication by $A_j$ is bounded from $H^{0,s}$ to $H^{0,s+\delta_2}$, $j=1,2,3$. We  have $P_j$ bounded from $H^{1,s}$ to $H^{0,s}$, and from $H^{0,s+\delta_2}$ to $H^{-1,s+\delta_2}$, $j=1,2,3$. Finally, multiplication by $\abs{A}^2$ is bounded from $H^{0,s}$ to $H^{0,s+2\delta_2}$, hence also from $H^{1,s}$ to $H^{-1,s+\delta_2}$. The result \eqref{Weps-map} then follows using
these mapping properties and the embedding results in \eqref{embed}.
\end{proof}

\begin{lem}\label{lemma23}
We have
 %----------------%
\begin{equation}\label{vanish}
\ip{\psi_0}{(P_jA_j+A_jP_j)\psi_0}=0,\quad j=1,2,3,
\end{equation}
 %----------------%
hence
 %----------------%
\begin{equation*}
\ip{\psi_0}{\We\psi_0}=\eps\ip{\psi_0}{\abs{A}^2\psi_0}.
\end{equation*}
 %----------------%
\end{lem}
\begin{proof}
The identities \eqref{vanish} follow from Assumption~\ref{assump}(\ref{assump3}), since we take $\psi_0$ to be real-valued.
\end{proof}

We introduce the notation
 %----------------%
\begin{equation}\label{YZdef}
\We=\eps^{1/2}Y+\eps Z,\quad Y=-A\cdot P-P\cdot A,\quad Z=\abs{A}^2.
\end{equation}
 %----------------%
together with
 %----------------%
\begin{equation}\label{b-def}
b=\ip{\psi_0}{\abs{A}^2\psi_0}=\ip{\psi_0}{Z\psi_0}.
\end{equation}
 %----------------%

In \cite{JKmag} a factorization of $\We$ is an essential tool. We recall its definition. Take as the intermediate space $\cK=L^2(\bR^3;\bC^9)$. First factor $A$ as follows:
 %----------------%
\begin{equation*}
A_j(s)=a_j(x)b_j(x),\quad a_j(x)=\abs{A_j(x)}^{1/2},
\quad b_j(x)=\sign(A_j(x))a_j(x),\quad j=1,2,3.
\end{equation*}
 %----------------%
We define
 %----------------%
\begin{equation*}
\we\colon H^{1,-\delta_2/2}\to \cK
\end{equation*}
 %----------------%
as a $9\times 1$ operator matrix. The entries are (omitting the column index)
 %----------------%
\begin{alignat*}{2}
(\we)_j&=\eps^{1/2}A_j,& j&=1,2,3,%\label{j123}
\\
(\we)_j&=\eps^{1/4}a_{j-3},& j&=4,5,6,%\label{j456}
\\
(\we)_j&=\eps^{1/4}b_{j-6}P_{j-6},\quad& j&=7,8,9.%\label{789}
\end{alignat*}
 %----------------%
Define the following self-adjoint and unitary operator on $\cK$:
 %----------------%
\begin{equation*}
\cU=\begin{bmatrix}
\bone_3 & 0 & 0\\
0 & 0& -\bone_3\\
0 & -\bone_3  & 0
\end{bmatrix},
\end{equation*}
 %----------------%
where the entries are $3\times3$ operator matrices, i.e.~the identity and the zero operator matrix.

Consider $\We$ as a bounded operator from $H^{1,-\delta_2/2}$ to $H^{-1,\delta_2/2}$. Using the above introduced quantities, we can write it in a factorized form
 %----------------%
\begin{equation}\label{factor}
\We=\wes\cU\we.
\end{equation}
 %----------------%

\begin{lem}\label{lemma24}
Let Assumption~\ref{assump}\textup{(\ref{assump2})} hold. Then for $0\leq s\leq \delta_2/2$ we have $\we\in\cB(H^{1,-s},\cK)$ and $\wes\in\cB(\cK,H^{-1,s})$, with the estimates
 %----------------%
\begin{equation}\label{weps-map}
\norm{\we}_{\cB(H^{1,-s},\cK)}\leq C\eps^{1/4}
\quad\text{and}\quad
\norm{\wes}_{\cB(\cK,H^{-1,s})}\leq C\eps^{1/4}.
\end{equation}
 %----------------%
\end{lem}

\begin{proof}
We prove the first estimate in \eqref{weps-map}, the second one follows by duality. By \eqref{decay} $(\we)_j$, $j=1,2,3$, are bounded from $H^{0,-\delta_2}$ to $H^{0,0}=L^2(\bR^3)$, hence also from $H^{1,-\delta_2/2}$ to $H^{0,0}$, due to \eqref{embed}. The norm is bounded by $C\eps^{1/2}$ for some $C>0$. Since $a_{j-3}(x)\lesssim\jap{x}^{-\delta_2/2}$, $j=4,5,6$, the terms $(\we)_j$, $j=4,5,6$, are bounded from $H^{1,-\delta_2/2}$ to $H^{0,0}$, with the norm bounded by $C\eps^{1/4}$. Since also $\abs{b_{j-6}(x)}\lesssim\jap{x}^{-\delta_2/2}$, $j=7,8,9$, and $P_{j-6}$, $j=7,8,9$ are bounded from $H^{1,-\delta_2/2}$ to $H^{0,-\delta_2/2}$, the mapping and norm results for these terms follow.
\end{proof}

%%%%%%%%%%%%%%%%%%%%%%%%%%%%%%
\section{Resolvent expansions}

We will need the asymptotic expansion of the resolvent $(H_0-z)^{-1}$ as $z\to0$ in the topology of $\cB(H^{-1,s},H^{1,-s'})$. We have the following result. As in \cite{JKmag} we change
to $\kappa$, where for $z\in\bC\setminus[0,\infty)$ we take $\kappa=-i\sqrt{z}$, with
$\im \sqrt{z}>0$, such that $z=-\kappa^2$.

\medskip

\begin{prop}\label{prop31a}
Let $N\geq-1$ be an integer. Assume that $V$ satisfies Assumptions~\ref{assump}\textup{(\ref{assump1})} and~\ref{assump}\textup{(\ref{assump3})} for some $\delta_1>2N+9$. Then we have
 %----------------%
\begin{equation}\label{eq31a}
(H_0+\kappa^2)^{-1}=\frac{1}{\kappa^2}G_{-2}+\frac{1}{\kappa}G_{-1}
+\sum_{j=0}^N\kappa^jG_j+\kappa^{N+1}\widetilde{G}_N(\kappa)
\quad \text{as $\kappa\to0$}
\end{equation}
 %----------------%
in the topology of $\cB(H^{-1,s},H^{1,-s'})$ for $s,s'>N+\frac92$. Here $G_{-2}=P_0=\ket{\psi_0}\bra{\psi_0}$ is the eigenprojection onto eigenvalue $0$ and $G_{-1}=c_0P_0$, where
 %----------------%
\begin{equation}\label{eq32a}
c_0=\frac{1}{12\pi}(X_1^2+X_2^2+X_3^2), \quad
X_j=\int_{\bR^3}\psi_0(x)V(x)x_jdx,
\quad j=1,2,3.
\end{equation}
 %----------------%
The error term satisfies $\norm{\widetilde{G}_N(\kappa)}_{\cB(H^{-1,s},H^{1,-s'})}\leq C$ for $\abs{\kappa}\leq \delta$, for some $\delta>0$. The constant $C$ depends on $s$, $s'$, and $\delta$.
If $j\geq-1$ is an odd integer, the coefficient $G_j$ is a finite rank operator.
\end{prop}

Comments on the proof of this proposition are postponed to Appendix~\ref{sectB}.

%%%%%%%%%%%%%%%%%%%%%%%%%%
\section{The SLFG formula}
\label{sec-SLFG}

We need the formula of many names, called SLFG in~\cite{JN}, for the Hamiltonians considered here.
As above $P_0$ is the eigenprojection for eigenvalue $0$ for $H_0$. Let $Q_0=I-P_0$. We write $\cH=P_0\cH\oplus Q_0\cH$. Let $R_{\eps}(z)=(H_{\eps}-zI)^{-1}$. We denote the resolvent of $Q_0H_{\eps}Q_0$ in the space $Q_0\cH$ by $R_{0,\eps}(z)$. i.e. $R_{0,\eps}(z)=(Q_0H_{\eps}Q_0-zQ_0)^{-1}$.

Then we have the matrix representation
 %----------------%
\begin{equation*}
R_{\eps}(z)=\begin{bmatrix}
\Reff(z) & R_{12}(z)\\ R_{21}(z) & R_{22}(z)
\end{bmatrix},
\end{equation*}
 %----------------%
where
 %----------------%
\begin{align*}
\Reff(z)&=\bigl(P_0H_{\eps}P_0-P_0\We Q_0R_{0,\eps}(z)Q_0\We P_0-zP_0\bigr)^{-1},
\\
R_{12}(z)&=-\Reff(z)P_0\We Q_0 R_{0,\eps}(z),
\\
R_{21}(z)&=-R_{0,\eps}(z)Q_0\We P_0\Reff(z),
\\
R_{22}(z)&=R_{0,\eps}(z)+R_{0,\eps}(z)
Q_0\We P_0\Reff(z)P_0\We Q_0R_{0,\eps}(z).
\end{align*}
 %----------------%

Let
 %----------------%
\begin{align}
F(z,\eps)&=\ip{\psi_0}{\We\psi_0}-z
-\ip{\psi_0}{\We Q_0 R_{0,\eps}(z)Q_0\We\psi_0}
\notag\\
&=\eps\ip{\psi_0}{\abs{A}^2\psi_0}-z
-\ip{\psi_0}{\We Q_0 R_{0,\eps}(z)Q_0\We\psi_0},
\label{F-eqn}
\end{align}
 %----------------%
where we used Lemma~\ref{lemma23} in the last step. Define $\Heff(z)$ on $P_0\cH$ via the equation
 %----------------%
\begin{equation} \label{eq-eff}
P_0(\Heff(z)-z)P_0=F(z,\eps)P_0.
\end{equation}
 %----------------%
Then we can write
 %----------------%
\begin{equation*}
\Reff(z)=(\Heff(z)-zP_0)^{-1}\quad\text{on $P_0\cH$}.
\end{equation*}
 %----------------%

We now use the factored resolvent equation for $(Q_0H_{0,\eps}Q_0-zQ_0)^{-1}$ and $(Q_0H_0Q_0-zQ_0)^{-1}$, using the factorization $Q_0\We Q_0=Q_0\wes\U\we Q_0$, see \eqref{factor}. We have on $Q_0\cH$
 %----------------%
\begin{align}
(Q_0H_{\eps}Q_0-zQ_0)^{-1}&=(Q_0H_0Q_0-zQ_0)^{-1}
\notag\\
&\quad-(Q_0H_0Q_0-zQ_0)^{-1}Q_0\wes
\notag\\
&\quad
\times\bigl[
\U+\we Q_0(Q_0H_0Q_0-zQ_0)^{-1}Q_0\wes
\bigr]^{-1}
\notag\\
&\quad\times\we Q_0
(Q_0H_0Q_0-zQ_0)^{-1}.
\label{FRE}
\end{align}
 %----------------%
We note that on $Q_0\cH$ one has
 %----------------%
\begin{equation*}
Q_0(H_0-z)^{-1}Q_0=\bigl(Q_0H_0Q_0-zQ_0\bigr)^{-1}
\end{equation*}
 %----------------%
and introduce the notation
 %----------------%
\begin{equation*}
G_{\eps}(z)=\we Q_0(H_0-z)^{-1}Q_0\wes.
\end{equation*}
 %----------------%
Then we can rewrite the last term in \eqref{F-eqn} using this notation and \eqref{FRE} to obtain
 %----------------%
\begin{align*}
\ip{\psi_0}{\We Q_0R_{0,\eps}(z)Q_0\We\psi_0}&=
\ip{\psi_0}{\wes\U\we Q_0R_{0,\eps}(z)Q_0\wes\U\we\psi_0}
\\
&=\ip{\psi_0}{\wes\U\bigl\{
G_{\eps}(z)-G_{\eps}(z)\bigl[\U+G_{\eps}(z)\bigr]^{-1}
G_{\eps}(z)
\bigr\}\U\we\psi_0}.
\end{align*}
 %----------------%
This leads to the following expression for $F(z,\eps)$:
 %----------------%
\begin{align}
F(z,\eps)&=\eps\ip{\psi_0}{\abs{A}^2\psi_0}-z
%\notag\\
%&\quad
-\ip{\psi_0}{\wes\cU\{G_{\eps}(z)-G_{\eps}(z)
[\cU+G_{\eps}(z)]^{-1}G_{\eps}(z)\}\cU\we\psi_0}.
\label{eq43}
\end{align}
 %----------------%

We would now like to use the resolvent expansion from Proposition~\ref{prop31a} to get an expansion of $G_{\eps}(z)$. There is an important point to be made. We cannot just apply $Q_0$ on both sides of the expansion of $(H_0-z)^{-1}$. The reason is that the eigenfunction $\psi_0$ may not decay fast enough to ensure that $Q_0H^{-1,s}\subset H^{-1,s}$ for the $s$ so large that the expansion holds. Instead, we observe that on the space $\cH$ we have
 %----------------%
\begin{equation*}
Q_0(H_0-z)^{-1}Q_0=(H_0-z)^{-1}+\frac{1}{z}P_0.
\end{equation*}
 %----------------%
The right-hand side has an asymptotic expansion in the topology $\cB(H^{-1,s},H^{1,-s})$, see Proposition~\ref{prop31a}. We now switch to the $\kappa$ variable instead of the $z$ variable.

\medskip

\begin{prop}\label{prop42}
Let $N\geq0$ be an integer. Let Assumption~\ref{assump} hold with $\delta_j>2N+9$, $j=1,2$. Then the following asymptotic expansion holds in the topology of $\cB(\cK)$, as $\kappa\to0$.
 %----------------%
\begin{equation*}%\label{Geps-expand}
G_{\eps}(\kappa)=\sum_{j=-1}^N\kappa^j\widetilde{G}_{\eps,j}+\kappa^{N+1}\widehat{E}_N(\kappa,\eps),
\end{equation*}
 %----------------%
where the coefficients,  $j=-1,0,\ldots,N$, are
 %----------------%
\begin{equation*}%\label{Geps-expand-c}
\widetilde{G}_{\eps,j}=\we G_j\wes,\quad \widehat{E}_N(\kappa,\eps)=\we\widetilde{G}_N(\kappa)\wes.
\end{equation*}
 %----------------%
with
 %----------------%
\begin{equation*}
\norm{\widetilde{G}_{\eps,j}}_{\cB(\cK)}\leq C_1\eps^{1/2},
\quad \norm{\widehat{E}_N(\kappa,\eps)}_{\cB(\cK)}\leq C_2\eps^{1/2},
\end{equation*}
 %----------------%
and $C_2$ independent of $\kappa$ for $\abs{\kappa}\leq1$.
\end{prop}

\begin{proof}
The results follow from Proposition~\ref{prop31a} and Lemma~\ref{lemma24}. % We omit the details.
\end{proof}

\begin{assumption}\label{nu-assump}
Let $\nu\geq-1$ be an odd integer. Assume that $G_j=0$ for $j=-1,1,\ldots,\nu-2$.
Assume that $\ip{\psi_0}{YG_{\nu}Y\psi_0}\neq0$, with $Y$ given in \eqref{YZdef}.
\end{assumption}

\medskip

\begin{rem}\label{finite-rank-remark}
This assumption is partly motivated by the result that
all expansion coefficients $G_j$ in \eqref{eq31a} with index $j\geq-1$ odd are finite rank operators. Furthermore, there exist models with  a different choice of $H_0$ such that this assumption can be verified for all $\nu\geq-1$ odd, see \cite[Theorem~5.9 and Remark~5.11]{JN}.
\end{rem}

\medskip

\begin{assumption}\label{pos-assump}
If $\nu=-1$, assume that
 %----------------%
\begin{equation}\label{assump-b}
 b-2\ip{\psi_0}{YG_0Y\psi_0}>0.
\end{equation}
 %----------------%
If $\nu\geq1$, assume that
 %----------------%
\begin{equation}\label{assump-bb}
 b-\ip{\psi_0}{YG_0Y\psi_0}>0.
\end{equation}
 %----------------%
\end{assumption}

Define
 %----------------%
\begin{equation}\label{tb-def}
\tb=b-\ip{\psi_0}{YG_0Y\psi_0},
\end{equation}
 %----------------%
with  $b$ given by \eqref{b-def}, and note that $\tb>0$. Furthermore, we introduce
 %----------------%
\begin{equation*}%\label{r-def}
r(\nu)=\begin{cases}
2, & \nu=-1,\\
4, & \nu=1,\\
\tfrac12(\nu+8), & \nu\geq3,
\end{cases}
\end{equation*}
 %----------------%
and define
 %----------------%
\begin{equation*}%\label{Deps}
\Deps=\{z=x+i\eta\mid
\tfrac12\tb\eps<x<\tfrac32\tb\eps,\;0<\abs{\eta}<(\tb\eps)^{r(\nu)}\}.
\end{equation*}
 %----------------%
We freely mix the notation $z$ and $\kappa$ for points in $\Deps$. Note that if $z\in \Deps$ then we have a lower bound $\abs{\kappa}=\abs{z}^{1/2}>C\eps^{1/2}$
where $C$ depends only on the parameter entering into the definition of $\Deps$. The power $r(\nu)$ is chosen to ensure that for $z\in\Deps$ the imaginary part $\im z$ can be absorbed in the error term in the expansion of $F(z,\eps)$, cf. Lemmata~\ref{lemma44}--\ref{lemma46}.

We write in statements below `for sufficiently small $\eps$' instead of the more precise formulation `there exists $\eps_0$ such that the statement holds for all $\eps$ such that $0<\eps<\eps_0$.'

We are now ready to state the essential decomposition and expansion results. Let
\begin{equation}\label{eq427}
\beta_{\nu}=\ip{\psi_0}{YG_{\nu}Y\psi_0}, \qquad \nu =1,3,5,\ldots\,.
\end{equation}

\medskip

\begin{lem}[Case $\nu=-1$]
\label{lemma44}
Let Assumption~\ref{nu-assump} hold for $\nu=-1$. For $\eps$ sufficiently small and $z\in \Deps$ we have
 %----------------%
\begin{equation*}%\label{F-nu=-1}
F(z,\eps)=H(z,\eps)+r(z,\eps)
\end{equation*}
 %----------------%
with
 %----------------%
\begin{equation*}
\sup_{\substack{0<\eps<\eps_0\\ z\in\Deps}}\eps^{-2}\abs{r(z,\eps)}<\infty.
\end{equation*}
 %----------------%
and
 %----------------%
\begin{equation}\label{eq412}
H(z,\eps)=\eps\tb-z-\eps^{3/2}\alpha_{-1}-
\kappa^{-1}\eps^2\beta_{-1}.
\end{equation}
 %----------------%
Here $\tb$ is defined in \eqref{tb-def} and $\alpha_{-1}, \beta_{-1}$ are real numbers given by
 %----------------%
\begin{equation*}%\label{beta-1-def}
\alpha_{-1}=\ip{\psi_0}{(YG_0Z+ZG_0Y-YG_0YG_0Y)\psi_0}
\end{equation*}
 %----------------%
and
 %----------------%
\begin{equation}\label{g-1-def}
\beta_{-1}=bc_0\bigl(b-2\ip{\psi_0}{YG_0Y\psi_0}\bigr),
\end{equation}
 %----------------%
where $c_0$ is defined in \eqref{eq32a}. By Assumption \ref{pos-assump} we have $\beta_{-1}>0$.
\end{lem}

\medskip

\begin{lem}[Case $\nu=1$]\label{lemma45}
For $\eps$ sufficiently small and $z\in \Deps$ we have
 %----------------%
\begin{equation*}%\label{F-nu=1}
F(z,\eps)=H(z,\eps)+r(z,\eps)
\end{equation*}
 %----------------%
with
 %----------------%
\begin{equation*}
\sup_{\substack{0<\eps<\eps_0\\ z\in\Deps}}\eps^{-4}\, \abs{r(z,\eps)}<\infty,
\end{equation*}
 %----------------%
where
 %----------------%
\begin{equation}\label{eq417}
H(z,\eps)=\eps(\tb-\kappa^2\alpha_2-\kappa^4\alpha_4)-z-\eps\kappa (\beta_1+\kappa^2\alpha_3+\kappa^4\alpha_5)
+\eps^{3/2}f(z,\eps)+\eps^{3/2}\kappa g(z,\eps).
\end{equation}
 %----------------%
Here $\tb$ are $\beta_1$ are given by \eqref{tb-def} and \eqref{eq427}, respectively, and
 %----------------%
\begin{align*}
%\beta_1 & =\ip{\psi_0}{YG_1Y\psi_0}  \label{beta-1}, \\[5pt]
\alpha_j & =\ip{\psi_0}{YG_jY\psi_0},\quad j=2,3,4,5,
\end{align*}
 %----------------%
and $f(z,\eps)$ and $g(z,\eps)$ are polynomials in $z$ with real coefficients. We have the following expressions:
 %----------------%
\begin{align}
f(z,\eps)&=c_{\frac32,0}+\kappa^2 c_{\frac32,2}+\kappa^4 c_{\frac32,4}
+\eps^{1/2}\bigl(c_{2,0} + \kappa^2 c_{2,2}\bigr)\notag\\
&\quad
+\eps\bigl(c_{\frac52,0}+\kappa^2 c_{\frac52,2}\bigr)
+\eps^{3/2}C_{3,0}+\eps^2 c_{\frac72,0}
\notag
\end{align}
 %----------------%
and
 %----------------%
\begin{equation*}
g(z,\eps)=c_{\frac32,1}+\kappa^2c_{\frac32,3}
+\eps^{1/2}\bigl(c_{2,1}+\kappa^2c_{2,3}\bigr)
+\eps c_{\frac52,1}+\eps^{3/2}c_{3,1}.
\end{equation*}
 %----------------%
\end{lem}

\medskip

\begin{lem}[Case $\nu\geq3$]
\label{lemma46}
Let Assumption~\ref{nu-assump} hold for some $\nu\geq3$. Then we can write
 %----------------%
\begin{equation*}
F(z,\eps)=H(z,\eps)+r(z,\eps),
\end{equation*}
 %----------------%
where
 %----------------%
\begin{equation*}
\sup_{\substack{0<\eps<\eps_0\\ z\in\Deps}}\eps^{-(\nu+8)/2}\abs{r(z,\eps)}<\infty.
\end{equation*}
 %----------------%
We have
 %----------------%
\begin{align}
H(z,\eps)&=\eps\tb-z-\eps\kappa^{\nu}\bigl[\beta_{\nu}+\kappa^2a_2(\eps)
+\kappa^4a_4(\eps)\bigr] +\eps^{3/2}f(z,\eps) + \eps\kappa^2g(z,\eps),
\label{eq425}
\end{align}
 %----------------%
with $\beta_\nu$ given by \eqref{eq427}. The terms $a_2(\eps)$ and $a_4(\eps)$ are polynomials in $\eps^{1/2}$ with real coefficients. Furthermore, $f(z,\eps)$ and $g(z,\eps)$ are polynomials in $\kappa^2$ and $\eps^{1/2}$ with real coefficients. We have
 %----------------%
\begin{equation*}
f(z,\eps)=f_0+\cO(\eps^{1/2})\quad\text{and}\quad
g(z,\eps)=g_0+\cO(\eps^{1/2}),
\end{equation*}
 %----------------%
for $z\in\Deps$ and $\eps\to0$.
\end{lem}

To prove these lemmata we need the following result.

\medskip

\begin{lem}\label{lemma47}
For sufficiently small $\eps$ and $z\in \Deps$ the operator $\cU+G_{\eps}(\kappa)$ is invertible in $\cB(\cK)$ with
$\norm{[\cU+G_{\eps}(\kappa)]^{-1}}$ uniformly bounded in $z$ and $\eps$. An asymptotic expansion of $[\cU+G_{\eps}(\kappa)]^{-1}$ can be obtained using the finite Neumann expansion.
\end{lem}

\begin{proof}
We give the proof in the case $\nu=-1$. The arguments in the other cases are straightforward. Write $\cU+G_{\eps}(\kappa)=\cU(I+\cU G_{\eps}(\kappa))$ and define $X=\kappa^{-1}\cU\we G_{-1}\wes$. Using Proposition~\ref{prop42} we have
 %----------------%
\begin{equation*}
\cU G_{\eps}(\kappa)=X+\cO(\eps^{1/2}).
\end{equation*}
 %----------------%
For $z\in \Deps$ we only get a uniform bound on $X$, but we need decay in $\eps$. Using  Lemma~\ref{lemma23} we obtain
 %----------------%
\begin{equation*}
X^2=\kappa^{-2}\cU\we G_{-1}\We G_{-1}\wes=
\kappa^{-2}c_0^2\cU\we P_0(\eps^{1/2}Y+\eps Z) P_0\wes
=\eps\kappa^{-2}c_0^2\cU\we P_0ZP_0\wes,
\end{equation*}
 %----------------%
where $c_0$ is given in \eqref{eq32a}. It follows that $\norm{X^2}_{\cB(\cK)}\leq C\eps^{1/2}$ for $z\in \Deps$.

Thus for $\eps$ sufficiently small $(I-X^2)^{-1}$ exists and can be computed using the Neumann series. Then we have
 %----------------%
\begin{equation*}
(I+X)^{-1}=(I-X)(I-X^2)^{-1}.
\end{equation*}
 %----------------%
This implies $(I+X)^{-1}=I-X+\cO(\eps^{1/2})$. Thus we have
 %----------------%
\begin{equation*}
(I+\cU G_{\eps}(\kappa))^{-1}=(I+X)^{-1}\bigl(I+\cO(\eps^{1/2})(I+X)^{-1}\bigr)^{-1}.
\end{equation*}
 %----------------%
The results follow from this expression and the Neumann series.
\end{proof}

\begin{proof}[\bf Proof of Lemma~\ref{lemma44}]
The starting point is \eqref{eq43}. The last term in \eqref{eq43} is split as follows
 %----------------%
\begin{align}
&\ip{\psi_0}{\wes\cU\{G_{\eps}(z)-G_{\eps}(z)
[\cU+G_{\eps}(z)]^{-1}G_{\eps}(z)\}\cU\we\psi_0}
=\ip{\psi_0}{\wes\cU G_{\eps}(z)\cU \we\psi_0}
\label{term1}
\\[6pt]
&\qquad\qquad\qquad\qquad\qquad-\ip{\psi_0}{\wes\cU G_{\eps}(z)
[\cU+G_{\eps}(z)]^{-1}G_{\eps}(z)\cU\we\psi_0}
\label{term2}
\end{align}
 %----------------%
We begin by analyzing the term on the right-hand side of \eqref{term1} which can be, using Proposition~\ref{prop42}, expressed as
 %----------------%
\begin{align}\label{expand_first}
\ip{\psi_0}{\wes\cU G_{\eps}(z)\cU \we\psi_0}
&=\sum_{j=-1}^N\kappa^j\ip{\psi_0}{\We G_j \We\psi_0}
\kappa^{N+1}\ip{\psi_0}{\We\widetilde{G}_N(\kappa,\eps)\We\psi_0}.
\end{align}
 %----------------%
We rewrite the $j=-1$ term of this sum using \eqref{YZdef}:
 %----------------%
\begin{align}
\ip{\psi_0}{\We G_{-1}\We\psi_0}
&=\ip{\psi_0}{(\eps^{1/2}Y+\eps Z) G_{-1}(\eps^{1/2}Y+\eps Z)\psi_0}
\notag\\[5pt]
&=\eps\ip{\psi_0}{YG_{-1}Y\psi_0}+
\eps^{3/2}\ip{\psi_0}{(ZG_{-1}Y+YG_{-1}Z)\psi_0}
\notag\\[5pt]
&\quad+\eps^2\ip{\psi_0}{ZG_{-1}Z\psi_0}.
%\label{last}
\notag
\end{align}
 %----------------%
Combining $G_{-1}=c_0P_0$ with Lemma~\ref{lemma23} we find that the only nonzero term on the right-hand side is the last one. Thus we have shown that
 %----------------%
\begin{equation*}
\ip{\psi_0}{\We G_{-1}\We\psi_0}=c_0b^2\eps^2,
\end{equation*}
 %----------------%
using $b=\ip{\psi_0}{\abs{A}^2\psi_0}=\ip{\psi_0}{Z\psi_0}$ and $c_0$, defined in \eqref{eq32a}. Taking into account the $\kappa^{-1}$ factor and $z\in \Deps$ it follows that the $j=-1$ contribution to the sum in~\eqref{expand_first} is of order~$\eps^{3/2}$.

Next we analyze the $j=0$ term in the expansion~\eqref{expand_first}. Proceeding as above we get
 %----------------%
\begin{align*}
\ip{\psi_0}{\We G_{0}\We\psi_0}&=\eps\ip{\psi_0}{YG_0Y\psi_0}
+\eps^{3/2}\bigl(
\ip{\psi_0}{YG_0Z\psi_0} + \ip{\psi_0}{ZG_0Y\psi_0}\bigr)
%\notag
\\
&\quad+\eps^2\ip{\psi_0}{ZG_0Z\psi_0}.
\end{align*}
 %----------------%
obtaining a term of order $\eps$ which must be taken together with the first term in \eqref{eq43}.

The terms in \eqref{expand_first} with $1\leq j\leq N$ are of order $\eps^{(2+j)/2}$. The error term is of order $\eps^{(N+3)/2}$.

Next we use Lemma~\ref{lemma47} to deal with the term \eqref{term2}. We rewrite it as follows:
 %----------------%
\begin{align*}
&\ip{\psi_0}{\wes\cU G_{\eps}(z)
[\cU+G_{\eps}(z)]^{-1}G_{\eps}(z)\cU\we\psi_0}=
\ip{\psi_0}{\wes\cU G_{\eps}(z)\cU G_{\eps}(z)\cU\we\psi_0}
\\
&\qquad\qquad\qquad\qquad\qquad -\ip{\psi_0}{\wes\cU G_{\eps}(z)\cU G_{\eps}(z)\cU G_{\eps}(z)\cU\we\psi_0}
+\mathcal{O}(\eps^2).
\end{align*}
 %----------------%
Consider the first term on the right-hand side and apply the expansion in Proposition~\ref{prop42}.
 %----------------%
\begin{align*}
\ip{\psi_0}{\wes\cU G_{\eps}(z)\cU G_{\eps}(z)\cU\we\psi_0}
&=
\kappa^{-2}\ip{\psi_0}{\We G_{-1}\We G_{-1}\We\psi_0}
\\
&\ +\kappa^{-1}\bigl(\ip{\psi_0}{\We G_{-1}\We G_{0}\We\psi_0}
%\\
%&\qquad
+\ip{\psi_0}{\We G_{0}\We G_{-1}\We\psi_0}\bigr)
\\
&\ +\kappa^0A_0(\eps)  +\cO(\kappa\eps^2),
\end{align*}
 %----------------%
where
 %----------------%
\begin{equation*}
A_0(\eps)= \ip{\psi_0}{\We G_{0}\We G_{0}\We\psi_0}+\ip{\psi_0}{\We G_{-1}\We G_1\We\psi_0}+\ip{\psi_0}{\We G_{1}\We G_{-1}\We\psi_0}\, .
\end{equation*}
 %----------------%
The $\kappa^{-2}$ term is simplified using \eqref{YZdef} and Lemma~\ref{lemma23}. The result is (recall $b=\ip{\psi_0}{Z\psi_0}$)
 %----------------%
\begin{equation*}
\kappa^{-2}\ip{\psi_0}{\We G_{-1}\We G_{-1}\We\psi_0}=
\kappa^{-2}c_0^2b^3\eps^3.
\end{equation*}
 %----------------%
In a similar manner the $\kappa^{-1}$ terms are rewritten as
 %----------------%
\begin{align*}
\kappa^{-1}c_0b\bigl[2\eps^2\ip{\psi_0}{YG_0Y\psi_0}
+\eps^{5/2}\bigl(\ip{\psi_0}{YG_0Z\psi_0}&+\ip{\psi_0}{ZG_0Y\psi_0}\bigr)
+2\eps^3\ip{\psi_0}{ZG_0Z\psi_0}\bigr].
\end{align*}
 %----------------%
Finally we also rewrite the $\kappa^{0}$ term as follows;
 %----------------%
\begin{align*}
A_0(\eps) &=
\eps^{3/2} \ip{\psi_0}{YG_0YG_0Y\psi_0}
+ \eps^2\bigl(\ip{\psi_0}{ZG_0YG_0Y\psi_0}
+\ip{\psi_0}{YG_0ZG_0Y\psi_0}
\\
&\ +\ip{\psi_0}{YG_0YG_0Z\psi_0}\bigr) +\eps^{5/2}\bigl(
\ip{\psi_0}{ZG_0ZG_0Y\psi_0}
+\ip{\psi_0}{YG_0ZG_0Z\psi_0}
\\
&\ +\ip{\psi_0}{ZG_0YG_0Z\psi_0}\bigr)+\eps^3\ip{\psi_0}{ZG_0ZG_0Z\psi_0}
\\
&\ +2\eps c_0 b\bigl[
\eps\ip{\psi_0}{YG_1Y\psi_0} + \eps^{3/2}(\ip{\psi_0}{YG_1Z\psi_0}
+\ip{\psi_0}{ZG_1Y\psi_0}
\bigr)
+\eps^2\ip{\psi_0}{ZG_1Z\psi_0}
\bigr].
\end{align*}
 %----------------%
The error term $\cO(\kappa\eps^2)$ follows by analyzing the structure of subsequent terms in the expansion.

The results in the lemma are a reorganized version of the above fairly detailed computations.
\end{proof}

\begin{rem}
Lemmata \ref{lemma44} and \ref{lemma45} are reminiscent to \cite[Lemma 3.3]{JN}, where	 similar results are obtained for a 
model in which $H_\eps$ is a zero order perturbation of $H_0$. The difference here is that the perturbation parameter is $\sqrt{\eps}$ and not $\eps$. 
In order to ensure the needed order of magnitude of $r(z,\eps)$, we thus have to take into account more terms in the expansion of  
$F(z,\eps)$, 
see equation \eqref{eq43}.
This makes the proof technically more involved. 
Similar comment applies to Lemma \ref{lemma51} and \cite[Lemma 3.4]{JN}.

\end{rem}

\begin{proof}[\bf Proof of Lemma~\ref{lemma45}]
The proof is very similar to the proof of Lemma~\ref{lemma44}. We omit the details.
\end{proof}

\begin{proof}[\bf Proof of Lemma~\ref{lemma46}]
We assume that $\nu\geq3$ is an odd integer such that Assumption~\ref{nu-assump} holds. Let us analyze how to get the right decomposition. The goal is to get an error term of order $\eps^{(\nu+8)/2}$.

We write
 %----------------%
\begin{equation*}
F(z,\eps)=\eps b-z+\Esf_1+\Esf_2,
\end{equation*}
 %----------------%
where
 %----------------%
\begin{equation*}
\Esf_1=-\ip{\psi_0}{\wes\cU G_{\eps}(z)\cU\we\psi_0}
\end{equation*}
 %----------------%
and
 %----------------%
\begin{equation*}
\Esf_2=\ip{\psi_0}{\wes\cU G_{\eps}(z)
[\cU+G_{\eps}(z)]^{-1}G_{\eps}(z)\cU\we\psi_0}.
\end{equation*}
 %----------------%
We have
 %----------------%
\begin{equation} \label{E1-exp}
\Esf_1=-\ip{\psi_0}{\We\Bigl(
\sum_{j=0}^N\kappa^jG_j\Bigr)\We\psi_0}
+\kappa^{N+1}\ip{\psi_0}{\We \widetilde{G}_N(\kappa)\We\psi_0}.
\end{equation}
 %----------------%
The second term should go into the error term $r(z,\eps)$. Each $\We$  contributes $\eps^{1/2}$ and $\kappa\sim\eps^{1/2}$. Thus
 %----------------%
\begin{equation*}
\kappa^{N+1}\eps\sim \eps^{1+(N+1)/2}.
\end{equation*}
 %----------------%
Then $1+(N+1)/2=(\nu+8)/2$ shows that we must take $N=\nu+5$ in \eqref{E1-exp}.

Due to the assumption on the $G_j$ we split the sum above as
 %----------------%
\begin{equation}\label{even-odd1}
\sum_{j=0}^N\kappa^jG_j=\sum_{k=0}^{(\nu+5)/2}\kappa^{2k}G_{2k}
+\kappa\sum_{l=(\nu-1)/2}^{(\nu+3)/2}\kappa^{2l}G_{2l+1}.
\end{equation}
 %----------------%
Due to the two $\We$ factors we get the powers $\eps$, $\eps^{3/2}$ and $\eps^2$. Multiplying out terms and reordering them we find
 %----------------%
\begin{align*}
\kappa\!\!\!\sum_{l=(\nu-1)/2}^{(\nu+3)/2}\kappa^{2l} \ip{\psi_0}{\We\, G_{2l+1}\,  \We\psi_0}
&= \kappa^{\nu}\eps\bigl[
(c_{1,\nu}+\kappa^2c_{1,\nu+2}+\kappa^4c_{1,\nu+4})
\\
&
\quad +\eps^{\frac12}(c_{\frac32,\nu}+\kappa^2c_{\frac32,\nu+2}
+\kappa^4c_{\frac32,\nu+4})+\eps(c_{2,\nu}+\kappa^2c_{2,\nu+2})
\bigr],
\end{align*}
 %----------------%
where one term was absorbed in the error term. We use the notation $c_{\alpha,\beta}$ for the coefficient to $\eps^{\alpha}\kappa^{\beta}$ in the expansion.

A similar computation for the even terms in \eqref{even-odd1} yields the following contribution:
 %----------------%
\begin{equation*}
\sum_{k=0}^{(\nu+5)/2}\kappa^{2k} \ip{\psi_0}{\We\, G_{2k}\,  \We\psi_0} =
\eps\bigl[
\sum_{k=0}^{(\nu+3)/2}\kappa^{2k}\bigl(c_{1,2k}+
\eps^{1/2}c_{\frac32,2k}+\eps c_{2,2k}\bigr)
+\kappa^{\nu+5}c_{1,\nu+5}
\bigr].
\end{equation*}	
 %----------------%
This finishes our analysis of $\Esf_1$.

Next we need to analyze the contributions from $\Esf_2$. For the term $[\cU+G_{\eps}(z)]^{-1}$ we use Lemma~\ref{lemma47}. We omit the details of the computation.  We get a contribution with the following structure:
 %----------------%
\begin{equation*}
\Esf_2 = \sum_{j=2}^{\nu+6}\, \sum_{k=j+1}^{\min\{2(j+1),\nu+7\}}
 \Biggl(
\sum_{\substack{l=0\\ \text{$l$ even}}}^{\nu+7-k}\eps^{k/2}\kappa^l d_{j,k,l}
+\sum_{\substack{l=\nu\\ \text{$l$ odd}}}^{\nu+7-k}
\eps^{k/2}\kappa^l d_{j,k,l}
\Biggr) + \mathcal{O}(\eps^{(\nu+8)/2})   \, .
\end{equation*}
 %----------------%
For the sum over $l$ odd we use the standard convention that $\sum_{l=m}^n\cdots=0$ for $n<m$.

Collecting and rearranging the terms we get the results stated in the lemma.
\end{proof}

%%%%%%%%%%%%%%%%%%%%%%
\section{Main results}
\label{sec-main}

In this section we use the above obtained Lemmata~\ref{lemma44}--\ref{lemma46} to prove the main results. We begin with the analysis of the term $H(z,\eps)$.
%We assume $\nu\geq3$.
 It is clear that the limits
 %----------------%
\begin{equation*}
H_{\pm}(x,\eps)=\lim_{\eta\downarrow0}H(x\pm i\eta,\eps)
\end{equation*}
 %----------------%
exist. We have a decomposition of the real and imaginary parts that can be written as
 %----------------%
\begin{equation*}
H_{\pm}(x,\eps)=R(x,\eps)\pm i I(x,\eps).
\end{equation*}
 %----------------%
Recall that $H(z,\eps)$ is given, for $\nu=-1, \nu=1$ and $\nu\geq 3$,  by \eqref{eq412}, \eqref{eq417} and \eqref{eq425} respectively. In either case we obtain
 %----------------%
\begin{equation*}
R(x,\eps)=\tb\eps -x +  \varphi_\nu(x,\eps),
% \eps^{3/2}f(x,\eps)+\eps x g(x,\eps).
\end{equation*}
 %----------------%
where $\varphi_\nu$ is a real-valued function of $x$ which satisfies
 %----------------%
\begin{equation} \label{fi-nu}
\sup_{x\in (\frac12\tb\eps,\frac32\tb\eps)} |\varphi_\nu| = \mathcal{O}(\eps^{3/2}) , \qquad \sup_{x\in (\frac12\tb\eps,\frac32\tb\eps)} |\partial_x \varphi_\nu| = \mathcal{O}(\eps), \qquad \nu=-1,1,3,\ldots\, .
\end{equation}
 %----------------%
It follows that for $\eps$ sufficiently small we have
 %----------------%
\begin{equation*}
R(\tfrac12\tb\eps, \eps)<0\quad\text{and}\quad R(\tfrac32\tb\eps, \eps)>0.
\end{equation*}
 %----------------%
We also have
 %----------------%
\begin{equation*}
\frac{dR}{dx}(x,\eps)=-1+\cO(\eps),\quad\eps\to0.
\end{equation*}
 %----------------%
Thus there exists a unique $x_0(\eps)\in(\frac12\tb\eps,\frac32\tb\eps)$ such that
 %----------------%
\begin{equation} \label{R-x0}
R(x_0(\eps),\eps)=0\quad\text{and}\quad x_0(\eps)=\tb\eps+\cO(\eps^{3/2}).
\end{equation}
 %----------------%
Define
 %----------------%
\begin{equation*}
\Gamma(\eps)=-I(\xe,\eps).
\end{equation*}
 %----------------%
In view of Assumption~\ref{nu-assump} and the fact that
 %----------------%
\begin{equation*}
I(x,\eps)=i^{\nu-1}x^{\nu/2}\eps\,  \beta_\nu +\cO(\eps^{2+\frac\nu2}),
\end{equation*}
 %----------------%
see \eqref{eq427}, it follows from \cite[(3.15)]{JN} and Assumption~\ref{nu-assump} that $\Gamma(\eps)>0$. More precisely, the coefficients $G_j$ in the expansion \eqref{eq31a} are 
self-adjoint, and 
$\im R(z,\eps)\geq0$ for $\im z\neq0$. Note also that due to the assumption
$\ip{\psi_0}{YG_{\nu}Y\psi_0}\neq0$ there exist constants $0<c_1<c_2$ such that
 %----------------%
\begin{equation}
\begin{aligned}\label{Ibound}
c_1\, \eps^{3/2} & \leq \abs{I(x,\eps)}\leq c_2\, \eps^{3/2} , \qquad\quad\ \nu =-1, 1,\\[5pt]
c_1\, \eps^{(\nu+2)/2} & \leq \abs{I(x,\eps)}\leq c_2\,  \eps^{(\nu+2)/2} ,  \qquad \nu\geq 3 ,
\end{aligned}
\end{equation}
 %----------------%
for all $x\in(\frac12\tb\eps,\frac32\tb\eps)$ and $\eps$ sufficiently small.

We need to fix a small interval $\Ie\subset(\frac12\tb\eps,\frac32\tb\eps)$. Define
 %----------------%
\begin{equation}\label{I-eps}
\Ie=\begin{cases}
[\xe-\frac14\tb\eps,\xe+\frac14\tb\eps], & \nu=-1,1,\\[8pt]
[\xe-\frac{\Gamma(\eps)}{\eps^2},\xe+\frac{\Gamma(\eps)}{\eps^2}], & \nu\geq3.
\end{cases}
\end{equation}
 %----------------%
Note that for all $\eps$ sufficiently small we have $\Ie\subset(\frac12\tb\eps,\frac32\tb\eps)$.

Our goal is to obtain information on
 %----------------%
\begin{equation}\label{Ageps}
A_{g_{\eps}}(t)=\ip{\psi_0}{e^{-itH_{\eps}}g_{\eps}(H_{\eps})\psi_0},
\end{equation}
 %----------------%
where $g_{\eps}$ is the characteristic function of the interval $I_{\eps}$. Using Stone's formula and the SLFG formula we have
 %----------------%
\begin{equation*}
A_{g_{\eps}}(t)=\lim_{\eta\downarrow0}\frac{1}{2\pi i}
\int_{I_{\eps}}e^{-itx}\Bigl(\frac{1}{F(x+i\eta,\eps)}-
\frac{1}{F(x-i\eta,\eps)}\Bigr)dx.
\end{equation*}
 %----------------%

\begin{lem}\label{lemma51}
For sufficiently small $\eps$ we have
 %----------------%
\begin{equation*}
\biggl\lvert
A_{g_{\eps}}(t)-\frac{1}{2\pi i}\int_{I_{\eps}}e^{-itx}
\Bigl(\frac{1}{H_+(x,\eps)}-
\frac{1}{H_-(x,\eps)}\Bigr)dx
\biggr\rvert
\leq C\eps^{q(\nu)}
\end{equation*}
 %----------------%
with
 %----------------%
\begin{equation*}
q(\nu)=\begin{cases}
\frac12, & \nu=-1,\\[5pt]
\frac52, & \nu=1,\\[5pt]
3, & \nu\geq3.
\end{cases}
\end{equation*}
 %----------------%
\end{lem}

\begin{proof}
Assume first that  $\nu\geq3$. Lemma~\ref{lemma46} implies that
 %----------------%
\begin{equation*}
\abs{F(z,\eps)}\geq\abs{H(z,\eps)}-C\eps^{(\nu+8)/2}
\end{equation*}
 %----------------%
for $\eps$ sufficiently small and $z\in\Deps$. Then the estimate \eqref{Ibound} implies that we get
 %----------------%
\begin{equation*}
\abs{F(z,\eps)}\geq\frac12\abs{H(z,\eps)}
\end{equation*}
 %----------------%
for all $\eps$ small enough, and in that case
 %----------------%
\begin{equation*}
\biggl\lvert
\frac{1}{2\pi i}
\int_{I_{\eps}}e^{-itx}\Bigl(\frac{1}{F(x\pm i\eta,\eps)}-
\frac{1}{H(x\pm i\eta,\eps)}\Bigr)dx
\biggr\rvert
\leq C\eps^{(\nu+8)/2}\int_{I_{\eps}}\frac{1}{\abs{H(x\pm i\eta,\eps)}^2}\ dx.
\end{equation*}
 %----------------%
Now we can take the limit $\eta\downarrow0$ to get
 %----------------%
\begin{equation*} %\label{H-L-1}
\biggl\lvert\lim_{\eta\downarrow0}
\frac{1}{2\pi i}
\int_{I_{\eps}}e^{-itx}\Bigl(\frac{1}{F(x\pm i\eta,\eps)}-
\frac{1}{H(x\pm i\eta,\eps)}\Bigr)dx
\biggr\rvert
\leq C\eps^{(\nu+8)/2}\int_{I_{\eps}}\frac{dx}{\abs{H_{\pm}(x,\eps)}^2}.
\end{equation*}
 %----------------%
We have due to \eqref{R-x0}
 %----------------%
\begin{equation*}
\abs{H_{\pm}(x,\eps)}^2=\abs{R(x,\eps)}^2+\abs{I(x,\eps)}^2
=\abs{R(x,\eps)-R(x_0(\eps),\eps)}^2+\abs{I(x,\eps)}^2.
\end{equation*}
 %----------------%
Using the expressions for $R(x,\eps)$ and $I(x,\eps)$, and the estimate \eqref{Ibound}, we arrive at
 %----------------%
\begin{equation}\label{Hpm-lower}
\abs{H_{\pm}(x,\eps)}^2\geq C\bigl(\abs{x-x_0(\eps)}^2+\eps^{\nu+2}\bigr)
\end{equation}
 %----------------%
for some $C<1$, see Appendix~\ref{sec-app} for details. Then we get
 %----------------%
\begin{align*}
\eps^{(\nu+8)/2}\int_{I_{\eps}}\frac{1}{\abs{H_{\pm}(x,\eps)}^2}\ dx&\leq
C\eps^3\int_{-\infty}^{\infty}
\frac{\eps^{(\nu+2)/2}}{(x-x_0(\eps))^2+\eps^{\nu+2}}\ dx
\\[5pt]
&
=C\eps^3\int_{-\infty}^{\infty}\frac{1}{1+x^2}\ dx=C\pi \eps^3.
\end{align*}
 %----------------%
Suppose next that  $\nu\in\{-1,1\}$. We proceed as in the case $\nu\geq 3$ and using Lemmata \ref{lemma44} and \ref{lemma45} we arrive at
 %----------------%
\begin{equation*} %\label{H-L-2}
\biggl\lvert\lim_{\eta\downarrow0}
\frac{1}{2\pi i}
\int_{I_{\eps}}e^{-itx}\Bigl(\frac{1}{F(x\pm i\eta,\eps)}-
\frac{1}{H(x\pm i\eta,\eps)}\Bigr)dx
\biggr\rvert
\leq C\, \eps^{3+\nu}\int_{I_{\eps}}\frac{dx}{\abs{H_{\pm}(x,\eps)}^2}.
\end{equation*}
 %----------------%
In view of \eqref{Ibound} estimate \eqref{Hpm-lower} is now replaced by
 %----------------%
\begin{equation*}%\label{Hpm-lower-2}
\abs{H_{\pm}(x,\eps)}^2\geq C\bigl(\abs{x-x_0(\eps)}^2+\eps^3\bigr).
\end{equation*}
 %----------------%
From this point we follow the analysis of the case $\nu\geq 3$ line by line and get
 %----------------%
$$
 \eps^{3+\nu}\int_{I_{\eps}}\frac{dx}{\abs{H_{\pm}(x,\eps)}^2} \, \leq\, C\,  \eps^{\frac 32+\nu}
$$
 %----------------%
as required.
\end{proof}

\begin{lem}%\label{lemma52}
We have the following estimates:
 %----------------%
\begin{equation}\label{eq511}
\Big\lvert
\frac{d}{dx}H(x,\eps)+1
\Big\rvert
\leq C\begin{cases}
\eps^{1/2}, & \nu=-1,1,\\
\eps, & \nu\geq3.
\end{cases}
\end{equation}
 %----------------%
\begin{equation*}%\label{eq512}
\Big\lvert
\frac{d^2}{dx^2}H(x,\eps)
\Big\rvert
\leq C\begin{cases}
\eps^{-1/2}, & \nu=-1,1,\\
\eps^{1/2}, & \nu=3,\\
\eps, & \nu\geq5.
\end{cases}
\end{equation*}
\end{lem}

\begin{proof}
Assume $\nu\geq3$. Then \eqref{eq425} in combination with \eqref{fi-nu} implies
 %----------------%
\begin{align*}
\frac{d}{dx}H(x,\eps)&=-1+\frac{\eps(-i)^{\nu}}{2} \bigl[\beta_{\nu}\, \nu\, x^{\frac{\nu-2}{2}}
-a_2(\eps)(\nu+2)\, x^{\frac \nu2}+a_4(\eps)(\nu+4)x^{\frac{\nu+4}{2}}\bigr]
 +  \mathcal{O}(\eps) .
\end{align*}
 %----------------%
This implies \eqref{eq511}. The second derivative is estimated in the same manner. The cases $\nu=-1$ and $\nu=1$ follow by using \eqref{eq412} and \eqref{eq417}, respectively.
\end{proof}

Next we define
 %----------------%
\begin{equation} \label{L-eq}
L_{\pm}(x,\eps)=-\bigl(x-x_0(\eps)\bigr)\pm iI(x_0(\eps),\eps).
\end{equation}
 %----------------%

\begin{lem}\label{lemma53}
For sufficiently small $\eps$ we have
 %----------------%
\begin{equation}
\label{estinteg}
\biggl\lvert
\int_{I_{\eps}}e^{-itx}\Bigl(\frac{1}{H_{\pm}(x,\eps)}-
\frac{1}{L_{\pm}(x,\eps)}\Bigr)dx
\biggr\rvert\leq C
\begin{cases}
\eps^{1/2}, & \nu=-1,1,\\[5pt]
\eps, & \nu\geq3.
\end{cases}
\end{equation}
 %----------------%
\end{lem}

\begin{proof}
Let us start with the case $\nu=-1$. It follows from the definition of $L_{\pm}(x,\eps)$ that we have $H_{\pm}(\xe,\eps)-L_{\pm}(\xe,\eps)=0$.

Using Taylor's formula with expansion point $x_0(\eps)$ we get the two estimates
 %----------------%
\begin{align}
\abs{H_{\pm}(x,\eps)-L_{\pm}(x,\eps)-
(1+\tfrac{d}{dx}H_{\pm}(x_0(\eps),\eps)(x-x_0(\eps))}&\leq C\eps^{-1/2}\abs{x-x_0(\eps)}^2,
\label{est-1}\\
\abs{H{_\pm}(x,\eps)-L_{\pm}(x,\eps)}&\leq C\eps^{1/2}\abs{x-x_0(\eps)}
\label{est-2}.
\end{align}
 %----------------%

We rewrite the bracket at the left-hand side of \eqref{estinteg} as follows
 %----------------%
\begin{equation}
\frac{1}{H_{\pm}(x,\eps)} - \frac{1}{L_{\pm}(x,\eps)}=
\frac{L_{\pm}(x,\eps)-H_{\pm}(x,\eps)}{L_{\pm}(x,\eps)^2}
+\frac{\bigl(L_{\pm}(x,\eps)-H_{\pm}(x,\eps)\bigr)^2}{L_{\pm}(x,\eps)^2H_{\pm}(x,\eps)}.
\label{est-3}
\end{equation}
 %----------------%
Note that \eqref{Hpm-lower} holds in the $\nu=-1$ case with $\eps^{\nu+2}$ replaced by $\eps^3$. Then we have
 %----------------%
\begin{equation*}
\abs{H_{\pm}(x,\eps)}\geq C\abs{x-x_0(\eps)}.
\end{equation*}
 %----------------%
Using this lower bound in combination with  \eqref{est-1}, \eqref{est-2}, and \eqref{est-3} it follows that
 %----------------%
\begin{align}
\biggl\lvert
\int_{I_{\eps}}e^{-itx} \Bigl(\frac{1}{H_{\pm}(x,\eps)}-
\frac{1}{L_{\pm}(x,\eps)}\Bigr) dx
\biggr\rvert
&\leq
\abs{1+\tfrac{d}{dx}H_{\pm}(x_0(\eps))} \ \biggl\lvert
\int_{I_{\eps}}e^{-itx}\frac{x_0(\eps)-x}{L_{\pm}(x,\eps)^2}\ dx\biggr\rvert
\label{eq-a}\\[7pt]
&\quad
+C\eps^{-1/2}\int_{I_{\eps}}\frac{\abs{x-x_0(\eps)}^2}{\abs{L_{\pm}(x,\eps)}^2}\
dx
\label{eq-b}\\[7pt]
&\quad
+C\eps\int_{I_{\eps}}\frac{\abs{x-x_0(\eps)}}{\abs{L_{\pm}(x,\eps)}^2}\ dx.
\label{eq-c}
\end{align}
 %----------------%
Recall the choice of $I_{\eps}$, see \eqref{I-eps}.   We have
 %----------------%
\begin{equation*}
\frac{\abs{x-x_0(\eps)}^2}{\abs{L_{\pm}(x,\eps)}^2}\leq1 .
\end{equation*}
 %----------------%
Since the length of $I_{\eps}$ equals $\frac12\tb\eps$, it follows that the term  \eqref{eq-b} is bounded by $C\eps^{1/2}$. The term \eqref{eq-c} is estimated using \eqref{L-eq} as follows
 %----------------%
\begin{equation*}
 \eps\int_{I_{\eps}}\frac{\abs{x-x_0(\eps)}}{\abs{L_{\pm}(x,\eps)}^2}\ dx \leq C
\eps\int_{I_{\eps}}\frac{1}{\sqrt{(x-\xe)^2+\eps^3}}\ dx
\leq C\eps\abs{\ln(\eps)}.
\end{equation*}
 %----------------%
The integrand on the right-hand side of \eqref{eq-a} is rewritten as
 %----------------%
\begin{equation*}
\frac{x_0(\eps)-x}{L_{\pm}(x,\eps)^2}=\frac{1}{L_{\pm}(x,\eps)}
\pm\frac{i\Gamma(\eps)}{L_{\pm}(x,\eps)^2},
\end{equation*}
 %----------------%
and thus
 %----------------%
\begin{align*}
\biggl\lvert
\int_{I_{\eps}}e^{-itx}\frac{x_0(\eps)-x}{L_{\pm}(x,\eps)^2}\ dx\biggr\rvert
\leq \biggl\lvert
\int_{I_{\eps}}e^{-itx}\frac{1}{L_{\pm}(x,\eps)}dx
\biggr\rvert
+
\int_{I_{\eps}}\frac{\Gamma(\eps)}{\abs{L_{\pm}(x,\eps)}^2}\ dx.
\end{align*}
 %----------------%
Using the definition of $L_{\pm}(x,\eps)$ we find
 %----------------%
\begin{equation*}
\int_{I(\eps)}\frac{\Gamma(\eps)}{\abs{L_{\pm}(x,\eps)}^2}\ dx
\leq \int_{\R}\frac{\Gamma(\eps)}{x^2+\Gamma(\eps)^2}\ dx=\pi.
\end{equation*}
 %----------------%
Now define
 %----------------%
\begin{equation} \label{length}
\ell(\eps)=\frac{\text{length}(I_{\eps})}{2}\, .
\end{equation}
 %----------------%
Then
 %----------------%
\begin{equation*}
\biggl\lvert
\int_{I_{\eps}}e^{-itx}\frac{1}{L_{\pm}(x,\eps)}dx
\biggr\rvert
=\biggl\lvert
\int_{-\ell(\eps)/\Gamma(\eps)}^{\ell(\eps)/\Gamma(\eps)}
e^{-it\Gamma(\eps)y}\frac{1}{y\mp i}dy
\biggr\rvert
\leq C.
\end{equation*}
 %----------------%
Here, in the last step, we have used \cite[Eq.~(3.56)]{JN}, see also \cite[Lem.~B.1]{BK}. Using the estimate \eqref{eq511}
the proof is finished in the case $\nu=-1$. The cases $\nu\geq1$ are treated in the same manner.
\end{proof}

Recall that  $\xe=\tb\eps+\cO(\eps^{3/2})$ and  $\Gamma(\eps)=-I(\xe,\eps)$. We have the following asymptotics of $\Gamma(\eps)$.
 %----------------%
\begin{equation}\label{gamma-asymp}
\Gamma(\eps)=
\begin{cases}
\tb^{-1/2}\beta_{-1}\, \eps^{3/2}+\cO(\eps^2), & \nu=-1,
\\[5pt]
\tb^{1/2}\beta_{1}\, \eps^{3/2}+\cO(\eps^{2}), & \nu=1,
\\[5pt]
-i^{\nu-1}\, \tb^{\nu/2}\, \eps^{(\nu+2)/2}\beta_{\nu}
+\cO(\eps^{(\nu+4)/2}), & \nu\geq3,
\end{cases}
\end{equation}
 %----------------%
where the coefficients   $\beta_{-1}, \beta_1$, and $\beta_{\nu}$  are given by \eqref{g-1-def} and  \eqref{eq427}.

Let
 %----------------%
\begin{equation}\label{p-eq}
p = \begin{cases}
1/2, & \nu=-1,1,\\[6pt]
1, & \nu\geq3.
\end{cases}
\end{equation}
 %----------------%

\medskip

\begin{lem}\label{lemma54}
For $\eps$ sufficiently small we have
 %----------------%
\begin{equation}\label{A-est}
\abs{A_{g_{\eps}}(t)-e^{-it(\xe-i\Gamma(\eps))}}
\leq C\, \eps^p\, .
\end{equation}
 %----------------%
\end{lem}

\begin{proof}
From  Lemmata \ref{lemma51} and \ref{lemma53} it follows that
 %----------------%
\begin{equation*}
A_{g_{\eps}}(t) = \frac{1}{2\pi i}\int_{I_{\eps}}e^{-itx}\Bigl[
\frac{1}{L_+(x,\eps)}-\frac{1}{L_-(x,\eps)}
\Bigr]dx + \mathcal{O}(\eps^p),
\end{equation*}
 %----------------%
with an error term uniform in $t>0$. We have
 %----------------%
\begin{equation*}%\label{L-comp}
\frac{1}{2\pi i}\int_{I_{\eps}}e^{-itx}\Bigl[
\frac{1}{L_+(x,\eps)}-\frac{1}{L_-(x,\eps)}
\Bigr]dx
=\frac{1}{\pi}\int_{I_{\eps}}e^{-itx}
\frac{\Gamma(\eps)}{(x-x_0(\eps))^2+\Gamma(\eps)^2}dx.
\end{equation*}
 %----------------%
Since $\ell(\eps)/\Gamma(\eps)\to \infty$ as $\eps\to 0$, equations \eqref{length} and \eqref{gamma-asymp} imply
 %----------------%
\begin{equation*}%\label{dif-comp}
\begin{aligned}
\Bigl\lvert\Bigl(\int_{\R}-\int_{I_{\eps}}\Bigr)
e^{-itx}
\frac{\Gamma(\eps)}{(x-x_0(\eps))^2+\Gamma(\eps)^2}dx\Bigr\rvert
&\leq
C \int_{\ell(\eps)}^{\infty}\frac{\Gamma(\eps)}{x^2+\Gamma(\eps)^2}dx \\[5pt]
&= C\int_{\ell(\eps)/\Gamma(\eps)}^\infty \frac{1}{1+t^2}\ dt
\leq  C\  \frac{\Gamma(\eps)}{\ell(\eps)} \\[5pt]
& \leq C  \eps^\nu  .
\end{aligned}
\end{equation*}
 %----------------%
By the residue theorem we have
 %----------------%
\begin{equation*}
\frac{1}{\pi}\int_{\R}e^{-itx}
\frac{\Gamma(\eps)}{(x-x_0(\eps))^2+\Gamma(\eps)^2}dx
=e^{-it(x_0(\eps)-i\Gamma(\eps)}.
\end{equation*}
 %----------------%
Estimate \eqref{A-est} thus follows from the definition of $I_{\eps}$, see \eqref{I-eps}, and \eqref{gamma-asymp}.
\end{proof}

We can now state our main result.

\begin{thm}\label{main}
Suppose that Assumptions~\ref{nu-assump} and~\ref{pos-assump} hold. Then for all $\lambda$ sufficiently small we have
 %----------------%
\begin{equation*}
\sup_{t>0} \abs{\ip{\psi_0}{\big[ e^{-itH_{\lambda}}-e^{-it(x_0(\lambda)-i\Gamma(\lambda))}\big] \psi_0} }
\leq C \lambda^{2p}\, ,
\end{equation*}
 %----------------%
where
 %----------------%
\begin{equation}\label{gamma-lambda}
\Gamma(\lambda)=
\begin{cases}
\tb^{\, -\frac12}\beta_{-1}\, |\lambda|^3+\cO(\lambda^4), &\  \nu=-1,
\\[5pt]
\tb^{\, \frac 12}\beta_{1}\,  |\lambda|^3+\cO(\lambda^{4}), & \ \nu=1,
\\[5pt]
-i^{\nu-1}\, \tb^{\, \frac \nu2}\,  |\lambda|^{\nu+2}\, \beta_{\nu}
+\cO(|\lambda|^{\nu+4}), &\  \nu\geq3,
\end{cases}
\end{equation}
 %----------------%
and
 %----------------%
\begin{equation} \label{x0-asymp}
x_0(\lambda)=\tb \lambda^2+\cO(|\lambda|^{3}).
\end{equation}
 %----------------%
Recall that $p$ is given by \eqref{p-eq}.
\end{thm}

\begin{proof}
Since the coefficients $b, \tb$ and $\beta_\nu, \nu=-1,2,3,\ldots,$ are even in $A$, the claim will follow from the estimate
 %----------------%
\begin{equation}\label{enough}
\sup_{t>0} \abs{\ip{\psi_0}{\big[ e^{-itH_{\eps}}-e^{-it(\xe-i\Gamma(\eps))}\big] \psi_0} }
\leq C\, \eps^p
\end{equation}
 %----------------%
upon setting $\lambda=\sqrt{\eps}$. With Lemma  \ref{lemma54} at hand, the proof of \eqref{enough} becomes a straightforward variation on the proof of \cite[Theorem~3.7]{JN}. Following the argument of Hunziker \cite{hun}, see also \cite{JN}, we first we apply equation
 %----------------%
\eqref{A-est} with $t=0$ to get $|A_{g_{\eps}}(0)-1|   \lesssim  \eps^p$. Consequently, by \eqref{Ageps},
\begin{equation*}  %\label{eq-hunz-1}
\| \big(1- g_\eps (H_\eps )\big)^{\frac 12}\, \psi_0\| \, \leq \, C  \eps^p.
\end{equation*}
 %----------------%
This implies
 %----------------%
\begin{align} %\label{eq-hunz-2}
\big | \LL \psi_0, \, e^{-itH_\eps}\, \psi_0\RR -A_{g_{\eps}}(t) \big | & = \big | \LL \big(1- g_\eps (H_\eps)\big)^{\frac 12}\,
\psi_0,\, e^{-it H_\eps}\big(1- g_\eps (H_\eps)  \big)^{\frac 12}\, \psi_0 \RR \big)  \nonumber \\
& \leq \| \big(1- g_\eps (H_\eps)  \big)^{\frac 12}\, \psi_0\|\, \leq\  C \eps^p.\nonumber
\end{align}
 %----------------%
Equation \eqref{enough} now follows from Lemma \ref{lemma54}. Since  $\Gamma(\eps)$ satisfies \eqref{gamma-asymp}, and 
since $\xe=\tb\eps+\cO(\eps^{3/2})$, cf.~equation \eqref{R-x0}, this completes the proof.
\end{proof}

\begin{rem}
Notice that the leading term in the asymptotic expansion of $\Gamma$ is of the same order of magnitude for $\nu=-1$ and $\nu=1$. This is yet another
important difference with respect to the resonances induced by an electric field, cf.~\cite[Thm.~3.7]{JN}.
\end{rem}

\begin{cor}\label{cor-unique}
Let the assumptions of Theorem \ref{main} be satisfied. Assume that there exist $\widehat x_0(\lambda)$ and $\widehat\Gamma(\lambda)$ such that for all $\lambda$ small enough we have
 %----------------%
$$
\sup_{t>0} \abs{\ip{\psi_0}{\big[ e^{-itH_{\lambda}}-e^{-it(\widehat x_0(\lambda)-i\widehat\Gamma(\lambda))}\big] \psi_0} }
\leq C\, \lambda^{2p}\, .
$$
 %----------------%
Then
 %----------------%
$$
\lim_{\lambda\to 0} \frac{\widehat\Gamma(\lambda)}{\Gamma(\lambda)} =\lim_{\lambda\to 0} \frac{\widehat x_0(\lambda)}{x_0(\lambda)} =1.
$$
 %----------------%
\end{cor}

\begin{proof}
From \eqref{gamma-asymp} and from \cite[Prop.~1.1]{JN2} it follows that
 %----------------%
$$
|\widehat\Gamma(\lambda)-\Gamma(\lambda)| + |\widehat x_0(\lambda)- x_0(\lambda)|  = \mathcal{O}(|\lambda|^{2p+2+|\nu|}) \qquad\lambda\to 0
$$
 %----------------%
for all $\nu=-1,1,3,\ldots\,.$ The claim thus follows from equations  \eqref{gamma-lambda} and \eqref{x0-asymp}.
\end{proof}

%%%%%%%%%%%%%%%%%%%%%%%%%%%%%%%%%%%%%%%%%%%%%%%%%%%%%
\section*{\bf Acknowledgement}
The authors thank the referees for helpful comments and suggestions.

%%%%%%%%%%%%%%%%%%%%%%%%%%%%%%%%%%%%%%%
\section*{Appendices}

\appendix
\section{Comments on \eqref{Hpm-lower}}
\label{sec-app}

The bound
 %----------------%
\begin{equation*}
\abs{H_{\pm}(x,\eps)}^2\geq C\bigl(\abs{x-x_0(\eps)}^2+\eps^{\nu+2}\bigr)
\end{equation*}
 %----------------%
may require some explanation. We give the details here. We have
 %----------------%
\begin{equation*}
H_{\pm}(x,\eps)=R(x,\eps)-R(\xe,\eps)\pm i I(\xe,\eps)
\end{equation*}
 %----------------%
with the first two terms on the right-hand side being
 %----------------%
\begin{multline*}
R(x,\eps)-R(\xe,\eps)
\\
=\xe-x+\eps^{3/2}[f(x,\eps)-f(\xe,\eps)]
+\xe\eps g(\xe,\eps) - x\eps g(x,\eps),
\end{multline*}
 %----------------%
where $f(x,\eps)$ is a polynomial in $x$. Thus we have
 %----------------%
\begin{equation*}
\eps^{3/2}[f(x,\eps)-f(\xe,\eps)]=\eps^{3/2}(x-\xe)[c+\cO(\eps)].
\end{equation*}
 %----------------%
Furthermore, $g(x,\eps)$ is also a polynomial in $x$, so that
 %----------------%
\begin{align*}
x\eps g(\xe,\eps) - \xe\eps g(x,\eps)&=\eps(x-\xe)g(x,\eps)
+\eps\xe[g(x,\eps)-g(\xe,\eps)]
\\
&=\eps(x-\xe)[c+\cO(\eps)].
\end{align*}
 %----------------%
which implies
 %----------------%
\begin{equation*}
\abs{R(x,\eps)-R(\xe,\eps)}^2=\abs{x-\xe}^2(1+\cO(\eps^2))
\geq\tfrac12\abs{x-\xe}^2
\end{equation*}
 %----------------%
for $\eps$ sufficiently small. Since $\beta_{\nu}\neq0$, we have
 %----------------%
\begin{equation*}
I(\xe,\eps)=i^{\nu-1}\eps(\xe)^{\nu/2}[\beta_{\nu}+\cO(\xe)],
\end{equation*}
 %----------------%
and therefore
 %----------------%
\begin{equation*}
\abs{I(\xe,\eps)}^2\geq C \eps^{\nu+2}.
\end{equation*}
 %----------------%
The estimate \eqref{Hpm-lower} follows from these computations.

%%%%%%%%%%%%%%%%%%%%%%%%%%%%%%%%%%%%%%%%%%%%%%%%%%%%%%%%%%%%
\section{Comments on the proof of Proposition~\ref{prop31a}}
\label{sectB}
Let us explain how Proposition~\ref{prop31a} can be obtained from several results available in the literature. The asymptotic expansion results are given in the topology of $\cB(-1,s;1,s')$.

The resolvent $R_0^0(\kappa)=(-\Delta+\kappa^2)^{-1}$ on $L^2(\bR^3)$ is an integral operator with the kernel $(4\pi\abs{x-y})^{-1}\exp(-\kappa\abs{x-y})$. It gives the asymptotic expansion of $R_0^0(\kappa)$. We recall the following results.

\smallskip

\begin{lem}[{\cite[Lemma 2.2]{JK}}]%\label{free-exp}
Assume $p\in\N_0$ and $s>p+\frac32$. Then
 %----------------%
\begin{equation*}
R_0^0(\kappa)=\sum_{j=0}^p\kappa^jG_j^0 + \cO(\kappa^{p+1})
\end{equation*}
 %----------------%
as $\kappa\to0$, $\re\kappa\geq0$, in $\cB(-1,s;1,-s)$. Here the operators $G_j^0$ are given by their integral kernels
 %----------------%
\begin{equation}\label{Gj-kernel}
G_j^0(x,y)=(-1)^j\frac{\abs{x-y}^{j-1}}{4\pi j!},\quad j\geq0.
\end{equation}
 %----------------%
We have
 %----------------%
\begin{equation*}%\label{G0-map}
G_0^0\in\cB(-1,s;1,-s')\quad\text{for $s,s'>\tfrac12$ and $s+s'\geq2$},
\end{equation*}
 %----------------%
and for $j\geq1$
 %----------------%
\begin{equation*}%\label{Gj-map}
G_j^0\in\cB(-1,s;1,-s')\quad\text{for $s,s'>j+\tfrac12$}.
\end{equation*}
 %----------------%
For $j$ odd the operator $G_j^0$ is a finite rank operator.
\end{lem}

Note that the finite rank property was not stated explicitly in \cite[Lemma 2.2]{JK}. It follows immediately from \eqref{Gj-kernel}, since the kernel of $G_j$ for $j$ odd is a polynomial in $x$ and $y$,see also  \cite[Proposition~7.1]{JN-RMP}.

To get the expansion of $R_0(\kappa)=(-\Delta+V+\kappa^2)^{-1}$ we use the following assumption on~$V$.

\begin{assumption}\label{assumpB}
Let $\delta>1$. Assume $v\colon H^{1,-\delta/2}\to \cK_0$ is a compact operator, and that that $\cV\in\cB(\cK_0)$ is a unitary and self-adjoint operator.
\end{assumption}

\noindent Under this assumption we define $V=v^{\ast}\cV v$, such that  $V\colon H^{1,-\delta/2}\to H^{-1,\delta/2}$ is a compact operator. $H_0=-\Delta+V$ is defined as the quadratic form sum. In particular, for
$V$ being a multiplication operator which satisfies satisfying Assumption~\ref{assump}(1) we can take $1<\delta<\delta_1$, $\cK_0=L^2(\bR^3)$, $v$ multiplication by $\abs{V(x)}^{1/2}$, and $\cV$ multiplication by $\sign V(x)$. Then Assumption~\ref{assumpB} is satisfied.

In the sequel we use Assumption~\ref{assumpB} for $\delta>1$ sufficiently large. The threshold analysis is based on the operator
 %----------------%
\begin{equation*}
M(\kappa)=\cV+v R_0^0(\kappa)v^{\ast}.
\end{equation*}
 %----------------%
which is a bounded operator on $\cK_0$. The general method is given in \cite{JN-RMP}. The $d=3$ case is analyzed in detail in \cite[Appendix~A]{JN}, using the topology of $\cB(0,s;0,-s')$. See also the computations in \cite[Section~IV]{JKmag}, for a different class of perturbations, but in the topology of $\cB(-1,s;1,-s')$ in intermediate steps.

The analysis of $M(\kappa)$ leads to the classification of the point $0$ in the spectrum of $H_0$ into the regular case, where $M(\kappa)$ is invertible for $\kappa$ small, and the three exceptional cases where there exist nontrivial solutions to $H_0\psi=0$, and the expansion of $M(\kappa)^{-1}$ will have singularities $\kappa^{-2}$ and/or $\kappa^{-1}$. Proposition~\ref{prop31a} concerns the second exceptional case, where $0$ is an eigenvalue of $H_0$, but not a zero resonance.

In all three exceptional cases we have an asymptotic expansion of $M(\kappa)^{-1}$ up to an order depending on $\delta$ in the assumption on $V$. The structure of the expansion is
 %----------------%
\begin{equation*}
M(\kappa)^{-1}=\sum_{j=-2}^N\kappa^jm_j+\cO(\kappa^{N+1}).
\end{equation*}
 %----------------%
This expansion is in the operator topology of $\cB(\cK_0)$. The coefficients $m_j$ are bounded self-adjoint operators on $\cK_0$.

To get the expansions of $R_0(\kappa)$ we use the factored resolvent equation
 %----------------%
\begin{equation}\label{fac-res}
R_0(\kappa)=R_0^0(\kappa)-R_0^0(\kappa)v^{\ast}
M(\kappa)^{-1}vR_0^0(\kappa).
\end{equation}
 %----------------%
The expansions of $R_0(\kappa)$ are obtained by inserting the expansions of $R_0^0(\kappa)$ and $M(\kappa)^{-1}$ into \eqref{fac-res} and multiplying out the three expansions in the second term, keeping terms up to the desired order $\cO(\kappa^{N})$.

The expression for the coefficient $G_{-1}$ in \eqref{eq32a} is given in \cite[Theorem~5.1]{JN}, with a proof in \cite[Appendix~A]{JN}.

Let us now explain how to obtain the finite rank property of the $G_j$ with odd index, as stated in the last part of Proposition~\ref{prop31a}. The arguments come from the proof of \cite[Proposition~7.1]{JN-RMP} and are here adapted to the present framework.

We have
 %----------------%
\begin{equation}\label{M-Mstar}
M(\kappa)^{-1}-(M(\kappa)^{-1})^{\ast}
=M(\kappa)^{-1}v\bigl[
(R_0^0(\kappa))^{\ast}-R_0^0(\kappa)
\bigr]v^{\ast}(M(\kappa)^{-1})^{\ast}.
\end{equation}
 %----------------%
Assume $\re\kappa=0$ and $\im\kappa\neq0$. Then $R_0^0(\kappa)-(R_0^0(\kappa))^{\ast}$ has an asymptotic expansion with all coefficients finite rank operators under the restriction on $\kappa$. Multiplying the asymptotic expansions and using uniqueness of coefficients in such expansions we conclude that the left-hand side of \eqref{M-Mstar} has an expansion with all coefficients of finite rank. Removing the restriction on $\kappa$ we conclude that the coefficients to odd powers of $\kappa$ in the expansion of $M(\kappa)^{-1}$ are of finite rank.

The asymptotic expansion of $R_0(\kappa)$ is obtained by inserting the expansions of $R_0^0(\kappa)$ and $M(\kappa)^{-1}$ in \eqref{fac-res} and multiplying the expansions in the second term. For the first term we already know that the coefficients to odd powers of $\kappa$ are of finite rank. For the product of the three expansions in the second term of \eqref{fac-res} a coefficient to an odd power of $\kappa$ is a sum of products of terms from the three expansions. To get an odd power we must have either one or three factors from the expansions with coefficients to odd powers of $\kappa$, hence the coefficient in the expansion of $R_0(\kappa)$ is a finite rank operator.

%%%%%%%%%%%%%%%%%%%%%%%%%%

\end{document}